\documentclass[11pt]{article}

\usepackage[margin=1in]{geometry}
\usepackage[T1]{fontenc}
\usepackage[utf8]{inputenc}
\usepackage{lmodern}
\usepackage{microtype}
\usepackage{setspace}
\usepackage{amsmath}
\usepackage{amssymb}
\usepackage{amsfonts}
\usepackage{amsthm}
\usepackage{mathtools}
\usepackage{bbm}
\usepackage{bm}

\newtheorem{theorem}{Theorem}[section]
\newtheorem{proposition}[theorem]{Proposition}

\theoremstyle{definition}

\newtheorem{assumption}[theorem]{Assumption}

\usepackage{booktabs}
\usepackage{array}
\usepackage{longtable}
\usepackage{multirow}
\usepackage{graphicx}
\usepackage{float}

\usepackage[authoryear,round]{natbib}
\usepackage{xcolor}
\usepackage[
colorlinks=true,
linkcolor=blue,
citecolor=blue,
urlcolor=blue
]{hyperref}
\usepackage[nameinlink,noabbrev]{cleveref}

\usepackage{enumitem}
\usepackage{xparse}

\newcommand{\E}{\mathbb{E}}
\newcommand{\Pp}{\mathbb{P}}

\newcommand{\Unif}{\operatorname{Unif}}
\newcommand{\BB}{\operatorname{BetaBinomial}}
\newcommand{\logit}{\operatorname{logit}}

\makeatletter
\RenewDocumentCommand{\title}{o m}{\gdef\@title{#2}}
\RenewDocumentCommand{\author}{s o m}{\gdef\@author{#3}}
\newcommand{\@affiliation}{}
\NewDocumentCommand{\affil}{s o m}{\gdef\@affiliation{#3}}

\renewcommand{\abstract}[1]{%
    \begin{center}
        \begin{minipage}{0.92\textwidth}
            \small
            \noindent\textbf{Abstract.}\quad #1
        \end{minipage}
    \end{center}
    \vspace{0.8em}
}
\newcommand{\keywords}[1]{%
    \begin{center}
        \begin{minipage}{0.92\textwidth}
            \small
            \noindent\textbf{Keywords:} #1
        \end{minipage}
    \end{center}
    \vspace{1em}
}

\renewcommand{\maketitle}{%
    \begin{center}
        {\LARGE\bfseries \@title\par}
        \vspace{1.2em}
        {\large \@author\par}
        \ifx\@affiliation\@empty\else
        \vspace{0.6em}
        {\small \@affiliation\par}
        \fi
    \end{center}
    \vspace{1.2em}
}
\makeatother

\begin{document}

\title{Marginal Persistence and Dynamic Copula Dependence in Sovereign Rating Migration Counts:\\
A Discrete Interval-Likelihood MAGMAR Analysis}
\author{Marina Palaisti}
\maketitle

\abstract{%
This paper develops an observed-data likelihood for applying moving-aggregate
modified autoregressive (MAGMAR) copula time-series models to discrete sovereign
rating-migration counts with time-varying exposure. An annual count identifies a
probability-integral-transform interval rather than a unique latent point, so the
likelihood integrates the latent process over the complete sequence of count
intervals. A guided sequential Monte Carlo implementation incorporates the
parameter-dependent stationary marginal adjustment of MAGMAR directly inside this
integration. In sovereign ratings from Fitch, Moody's, and S\&P over 1963--2025,
the migration counts are strongly overdispersed relative to a binomial margin and,
under an exposure-conditioned static beta-binomial margin, MAGMAR improves
substantially on MAG. The interpretation changes
once the marginal mean is allowed to depend on the lagged migration rate: residual
serial correlation in the count margin disappears, a reference dynamic-margin fit
gives only a small MAGMAR improvement, and information criteria favor MAG. Numerical
profiles, higher-precision refits, and parameter-recovery experiments show that
marginal persistence and copula persistence are only weakly separated in this short
discrete series. The results therefore support interval integration as the correct
observation layer while showing that conclusions about latent dynamic dependence
must be conditioned on the specification of the count margin.}


\section{Introduction}
\label{sec:intro}

Sovereign rating migrations affect bond portfolios, bank and insurance balance
sheets, sovereign-backed lending, eligibility rules, and migration-based measures
of credit risk. Most quantitative work models the probability or intensity of a
transition between rating states. The present paper asks a different time-series
question: after accounting for the number of country--agency relationships exposed
to a possible migration, does aggregate annual migration activity display persistent
dynamic dependence?

Two features make this question nonstandard. First, the annual number of eligible
country--agency pairs changes sharply through time. A count of 20 migrations has a
different interpretation when 60 pairs are at risk than when 300 are at risk.
Second, migration counts are discrete. If the conditional count distribution in
year $t$ has CDF $F_t$, observing $A_t=a_t$ implies only that a latent uniform
variable lies in
\[
\bigl(F_t(a_t^-),F_t(a_t)\bigr],
\]
rather than revealing a unique probability-integral-transform value. A continuous
copula likelihood evaluated at one jittered point inside this interval is therefore
not the observed-data likelihood of the count sequence.

The methodological contribution of the paper is to construct an explicit discrete
observation layer for the MAG and MAGMAR copula processes of
\citet{pappert2026mag}. The latent recursion is retained, but the observation map is
formulated through exposure-conditioned count intervals. For each candidate
parameter vector, the stationary marginal CDF of the raw MAGMAR state is estimated,
the observed count intervals are mapped back to raw-state intervals, and a guided
particle filter integrates the conditional latent density over those intervals.
This produces a simulated approximation to the observed-count rectangle likelihood.

The empirical contribution is to separate two forms of persistence that can look
similar in a short discrete series. Under a static beta-binomial margin, Gaussian
MAGMAR is much better than Gaussian MAG: the publication fit gives log likelihoods
$-170.69$ and $-178.88$, respectively. The fitted MAGMAR dependence is characterized
by a strong positive autoregressive copula component and a negative moving-aggregate
component. Yet the static count margin leaves substantial residual serial
correlation. Once the beta-binomial mean is allowed to depend on the previous year's
migration rate, the residual Ljung--Box statistic falls from 102.43 to 10.04 at lag
10, with nominal $p=0.437$ for the dynamic margin. Under this richer observation
model, a reference joint fit yields only $0.77$ log-likelihood units of improvement
from MAGMAR over MAG, while AIC and BIC favor MAG.

The dynamic-margin results also reveal a finite-sample identification issue that is
substantively important rather than merely computational. Independent profile
likelihoods rise toward large positive values of the MAGMAR autoregressive parameter,
and higher-precision refits can move from a moderate-$\phi$ solution to a
high-$\phi$ solution. Parameter-recovery simulations reproduce the margin
coefficients well but recover $\phi$ and $\theta$ only weakly. Low-exposure
sensitivities likewise do not produce a stable MAGMAR advantage. The empirical
message is therefore not that one dependence specification dominates under every
margin. It is that discrete interval integration changes the inferential object,
and that marginal dynamics must be modeled carefully before persistent copula
dependence is interpreted structurally.

The paper makes four contributions:
\begin{enumerate}[leftmargin=2em]
    \item it formulates an exposure-conditioned discrete observation model for
    MAG/MAGMAR in which counts identify latent probability intervals;
    \item it develops a guided sequential Monte Carlo approximation to the resulting
    observed-count likelihood, including the parameter-dependent stationary marginal
    adjustment inside the likelihood calculation;
    \item it compares static and dynamic beta-binomial observation layers and shows
    how marginal persistence changes the apparent evidence for latent dependence;
    \item it provides numerical diagnostics for Monte Carlo error, profile geometry,
    probability-grid resolution, finite-sample parameter recovery, influential
    years, and direction-specific behavior.
\end{enumerate}

The remainder of the paper is organized as follows. Section~\ref{sec:literature}
reviews the related literature. Section~\ref{sec:data} describes the sovereign
ratings, exposure, and beta-binomial observation models. Section~\ref{sec:model}
develops the discrete MAG/MAGMAR likelihood. Section~\ref{sec:estimation} describes
simulation-based estimation and the validation design. Section~\ref{sec:results}
reports the empirical results, Section~\ref{sec:discussion} discusses interpretation,
and Section~\ref{sec:conclusion} concludes.

\section{Related literature and methodological background}
\label{sec:literature}

This paper connects three literatures: credit-rating transitions, copula models for
discrete time series, and simulation-based likelihood evaluation.

Rating-transition models commonly use discrete- or continuous-time Markov
structures. \citet{jarrow1997} develop a Markov framework for credit-risk spreads,
while \citet{lando2002} study rating transitions and rating drift using continuous
observations. \citet{bangia2002} show that migration behavior changes with the
business cycle, and \citet{figlewski2012} analyze macroeconomic effects on default
and rating-transition dynamics. Those models are designed around transition
probabilities, intensities, or transition matrices. Here the empirical object is the
aggregate number of annual sovereign rating migrations, modeled as a univariate
discrete time series conditional on the number of rating relationships exposed to a
possible migration.

Copula models separate marginal behavior from dependence through Sklar's theorem
\citep{sklar1959}; see \citet{joe2014} for a comprehensive treatment. This
separation is attractive when the marginal distribution changes mechanically with
exposure. Discrete margins, however, require care. For integer-valued observations
the copula is observed only on the range of the marginal CDF, so the latent uniform
is interval-censored rather than observed exactly. \citet{genestneslehova2007}
discuss the resulting identification issues for count data, while
\citet{trivedizimmer2017} emphasize that covariate variation, and hence variation in
observable CDF grids, can help distinguish parametric dependence structures.
\citet{weiss2018intcount} provides a broader treatment of discrete-valued time
series.

The latent dynamic model is the MAGMAR copula process of
\citet{pappert2026mag}. MAGMAR augments an autoregressive copula recursion with a
moving-aggregate component driven by current and lagged innovations. A key feature
for the present application is that the stationary marginal distribution of the raw
MAGMAR process need not be uniform. The stationary marginal transformation that
restores uniformity depends on the candidate dependence parameters. In a discrete
observation model this transformation therefore changes the raw-state intervals
associated with each observed count and must be recomputed inside the likelihood.

The resulting observed-data likelihood is a high-dimensional rectangle probability
for a nonlinear latent process and is not available in closed form. We evaluate it
sequentially using particle importance sampling and systematic resampling, following
the general SMC logic developed for nonlinear and non-Gaussian state-space models;
see \citet{delmoral2004} and \citet{doucetjohansen2011}. Common random numbers are
used within an optimization run so that the simulated objective is reproducible and
substantially smoother than an objective based on independent simulation at every
parameter evaluation.

\section{Data and the discrete observation layer}
\label{sec:data}

\subsection{Sovereign ratings and annual migration counts}

The empirical application uses sovereign credit ratings from Fitch, Moody's, and
S\&P obtained from TheGlobalEconomy. Ratings are standardized to an ordered numeric
scale on which larger values denote stronger credit quality. The raw information is
converted to a country--agency--year panel. For each country--agency pair, the annual
rating is compared with its preceding eligible annual rating. A downgrade occurs
when the current rating is lower, an upgrade when it is higher, and a migration when
either event occurs.

The modeled observation is the number of eligible country--agency pairs whose annual
rating differs from the preceding annual rating, rather than the number of every
within-year rating announcement. The construction starts from 9,748 raw rating rows;
187 NR/WR/SD/RD rows and 29 rows with unmapped rating codes are excluded. The
resulting annual panel contains 10,290 country--agency--year rows and 459
country--agency pairs. There are 9,831 eligible annual comparisons and 2,109
migrations, comprising 952 downgrades and 1,157 upgrades.

Let $E_t$ denote the number of country--agency pairs eligible for a migration
comparison in year $t$. Let $N_t^{-}$ and $N_t^{+}$ denote the numbers of downgrades
and upgrades, respectively, and define
\begin{equation}
A_t=N_t^{-}+N_t^{+}.
\label{eq:At}
\end{equation}
Pairs with fewer than three distinct annual observations are removed before
aggregation, so the minimum-history rule is applied at the country--agency level.
The aggregate series contains 63 annual observations from 1963 through 2025. The
migration count ranges from 0 to 101 and exposure from 1 to 390.

\begin{table}[t]
\centering
\caption{Construction audit and aggregate sovereign migration sample.}
\label{tab:data-summary}
\begin{tabular}{lr}
\toprule
Quantity & Value \\
\midrule
Raw rating rows & 9,748 \\
Excluded NR/WR/SD/RD rows & 187 \\
Rows with unmapped rating codes & 29 \\
Country--agency--year rows in annual panel & 10,290 \\
Country--agency pairs & 459 \\
Eligible annual comparison rows & 9,831 \\
Total migrations & 2,109 \\
Downgrades & 952 \\
Upgrades & 1,157 \\
Aggregate sample & 1963--2025 \\
Number of annual observations & 63 \\
Range of migration count $A_t$ & 0--101 \\
Range of exposure $E_t$ & 1--390 \\
\bottomrule
\end{tabular}
\end{table}

\subsection{Static and dynamic beta-binomial margins}
\label{subsec:bbmargin}

Because $A_t$ counts migrations among $E_t$ eligible pairs,
\[
A_t\in\{0,1,\ldots,E_t\}.
\]
For a generic exposure-conditioned discrete margin write
\[
F_t(a;\gamma)=\Pp_\gamma(A_t\le a\mid E_t,\mathcal F_{t-1}),
\qquad
m_t(a;\gamma)=\Pp_\gamma(A_t=a\mid E_t,\mathcal F_{t-1}).
\]

The observation layer is beta-binomial. Conditional on a migration probability
$p_t\in(0,1)$ and precision parameter $\kappa>0$,
\begin{equation}
A_t\mid E_t,\mathcal F_{t-1}
\sim
\BB(E_t,p_t,\kappa),
\label{eq:bb-margin}
\end{equation}
with beta shape parameters $\alpha_t=p_t\kappa$ and
$\beta_t=(1-p_t)\kappa$. Its probability mass is
\begin{equation}
m_t(a;p_t,\kappa)
=
\binom{E_t}{a}
\frac{B(a+p_t\kappa,E_t-a+(1-p_t)\kappa)}
     {B(p_t\kappa,(1-p_t)\kappa)},
\label{eq:bb-pmf}
\end{equation}
and
\begin{align}
\E(A_t\mid E_t,\mathcal F_{t-1})&=E_t p_t,\\
\operatorname{Var}(A_t\mid E_t,\mathcal F_{t-1})
&=E_t p_t(1-p_t)\frac{E_t+\kappa}{1+\kappa}.
\label{eq:bb-var}
\end{align}

Two specifications are used. The \emph{static} beta-binomial margin sets
\begin{equation}
\logit(p_t)=\beta_0.
\label{eq:static-margin}
\end{equation}
The \emph{dynamic} margin uses the lagged observed migration rate,
\begin{equation}
\logit(p_t)=\beta_0+\beta_1 R_{t-1},
\qquad
R_{t-1}=\frac{A_{t-1}}{E_{t-1}}.
\label{eq:dynamic-margin}
\end{equation}
The dynamic analysis therefore uses 62 observations from 1964--2025. Setting
$\beta_1=0$ recovers the static mean specification on the same sample.

For every admissible count define the CDF interval
\begin{align}
\ell_t(a;\gamma)&:=F_t(a^-;\gamma),\\
r_t(a;\gamma)&:=F_t(a;\gamma),\\
I_t(a;\gamma)&:=\bigl(\ell_t(a;\gamma),r_t(a;\gamma)\bigr].
\label{eq:count-interval}
\end{align}
Its width is exactly the probability mass,
\begin{equation}
|I_t(a;\gamma)|=m_t(a;\gamma).
\label{eq:interval-width}
\end{equation}
Thus an observed count identifies a region of latent probability space rather than
a single latent point.

\begin{proposition}[Exposure-conditional distributional transform]
\label{prop:dt}
Suppose $A_t\mid E_t,\mathcal F_{t-1}$ has CDF $F_t(\cdot;\gamma_0)$ and let
$V_t\sim\Unif(0,1)$ be independent of the count process. Then
\begin{equation}
U_t^{(V)}
=
F_t(A_t^-;\gamma_0)
+
V_t\,m_t(A_t;\gamma_0)
\label{eq:randomized-dt}
\end{equation}
is conditionally uniform on $(0,1)$ and satisfies
$A_t=F_t^{-1}(U_t^{(V)};\gamma_0)$ almost surely.
\end{proposition}

\begin{proof}
Conditional on the observed conditioning variables, the intervals
$I_t(a;\gamma_0)$ partition $(0,1]$ and have lengths equal to the corresponding
conditional count probabilities. Conditional on $A_t=a$, expression
\eqref{eq:randomized-dt} is uniform on $I_t(a;\gamma_0)$. Mixing these conditional
uniform distributions with weights $m_t(a;\gamma_0)$ therefore produces a uniform
distribution on $(0,1)$. The inverse-CDF statement follows from
$U_t^{(V)}\in I_t(A_t;\gamma_0)$ almost surely.
\end{proof}

The distributional transform is useful for diagnostics, but the observed-data
likelihood below integrates over the full interval rather than conditioning on one
realization of $V_t$.

\section{Discrete MAG and MAGMAR interval likelihood}
\label{sec:model}

\subsection{Gaussian MAG(1) and MAGMAR(1,1)}

Let $\{W_t\}_{t\in\mathbb Z}$ be iid $\Unif(0,1)$. For a bivariate copula
$C_\rho$ define
\[
h_\rho(u,v)=\frac{\partial}{\partial u}C_\rho(u,v),
\]
and let $h_\rho^{-1}(\cdot,v)$ denote inversion in the first argument. The empirical
implementation uses Gaussian copulas for the autoregressive and moving-aggregate
components.

The Gaussian MAG(1) process is
\begin{equation}
U_t=h_\theta^{-1}(W_t,W_{t-1}),
\label{eq:mag}
\end{equation}
where $\theta\in(-1,1)$. The Gaussian MAGMAR(1,1) process is
\begin{equation}
U_t
=
h_\phi^{-1}\!\left(
 h_\theta^{-1}(W_t,W_{t-1}),
 U_{t-1}
\right),
\label{eq:magmar}
\end{equation}
where $\phi\in(-1,1)$ is the autoregressive copula parameter and
$\theta\in(-1,1)$ is the moving-aggregate parameter. Setting $\phi=0$ reduces
MAGMAR(1,1) to MAG(1).

The conditional density contribution of the raw Gaussian MAGMAR state is
\begin{equation}
f_{\phi,\theta}(u_t\mid u_{t-1},w_{t-1})
=
c_\phi(u_t,u_{t-1})\,
c_\theta\!\left(h_\phi(u_t,u_{t-1}),w_{t-1}\right),
\label{eq:magmar-density}
\end{equation}
where $c_\rho$ is the Gaussian copula density. The innovation state is updated by
\[
w_t=h_\theta\!\left(h_\phi(u_t,u_{t-1}),w_{t-1}\right).
\]
For MAG(1), the density contribution reduces to $c_\theta(u_t,w_{t-1})$ and
$w_t=h_\theta(u_t,w_{t-1})$.

\subsection{Stationary marginal adjustment}

A raw MAGMAR process generated by \eqref{eq:magmar} need not have a uniform
stationary marginal \citep{pappert2026mag}. Let
\begin{equation}
\Psi_{\eta}(u)=\Pp_\eta(U_t\le u),
\qquad
\eta=(\phi,\theta)',
\label{eq:psi}
\end{equation}
denote the stationary marginal CDF of the raw latent process. When $\Psi_\eta$ is
continuous,
\begin{equation}
\widetilde U_t=\Psi_\eta(U_t)
\label{eq:uniform-adjust}
\end{equation}
is marginally uniform. The observed discrete process is defined by
\begin{equation}
A_t=F_t^{-1}(\widetilde U_t;\gamma).
\label{eq:obs-map}
\end{equation}
Thus observing $A_t=a_t$ is equivalent to
$\widetilde U_t\in I_t(a_t;\gamma)$ and, in the raw MAGMAR state,
\begin{equation}
U_t\in
J_t(a_t;\eta,\gamma)
:=
\left(
\Psi_\eta^{-1}\{\ell_t(a_t;\gamma)\},
\Psi_\eta^{-1}\{r_t(a_t;\gamma)\}
\right].
\label{eq:raw-interval}
\end{equation}
The inverse stationary transformation is therefore part of every likelihood
evaluation.

\subsection{Observed-count rectangle likelihood}
\label{subsec:exact-likelihood}

Let $\xi=(\eta',\gamma')'$ collect dependence and margin parameters and define
\[
\mathcal I_T(a_{1:T};\gamma)=\prod_{t=1}^T I_t(a_t;\gamma).
\]
If $\widetilde g_{\eta,T}$ is the joint density of
$(\widetilde U_1,\ldots,\widetilde U_T)$, the observed-data likelihood is
\begin{equation}
L_T^D(\xi)
=
\int_{\mathcal I_T(a_{1:T};\gamma)}
\widetilde g_{\eta,T}(u_{1:T})\,du_{1:T}.
\label{eq:exact-lik}
\end{equation}
Equivalently, this is the probability under the raw MAGMAR process that all $U_t$
fall inside the parameter-dependent intervals in \eqref{eq:raw-interval}.

\begin{theorem}[Randomized integral representation]
\label{thm:randomized-integral}
Suppose $\widetilde g_{\eta,T}$ exists and is integrable. For iid
$V_t\sim\Unif(0,1)$ define
\[
u_t(V_t;\gamma)=\ell_t(a_t;\gamma)+V_t m_t(a_t;\gamma).
\]
Then
\begin{equation}
L_T^D(\eta,\gamma)
=
\left\{\prod_{t=1}^T m_t(a_t;\gamma)\right\}
\E_V\!\left[
\widetilde g_{\eta,T}
\{u_1(V_1;\gamma),\ldots,u_T(V_T;\gamma)\}
\right].
\label{eq:randomized-integral}
\end{equation}
\end{theorem}

\begin{proof}
Apply the coordinatewise change of variables
$u_t=\ell_t(a_t;\gamma)+v_t m_t(a_t;\gamma)$ to the rectangle integral
\eqref{eq:exact-lik}. The Jacobian determinant is
$\prod_{t=1}^T m_t(a_t;\gamma)$.
\end{proof}

The order of averaging and taking logarithms matters:
\[
\log \E_V\{\widetilde g_{\eta,T}(u(V;\gamma))\}
\neq
\E_V\{\log\widetilde g_{\eta,T}(u(V;\gamma))\}.
\]
Hence a continuous likelihood evaluated at one randomized transform is a different
criterion from the observed-count likelihood in \eqref{eq:exact-lik}.

\subsection{Identification and changing exposure}
\label{subsec:identification}

For a discrete margin define the probability grid
\[
\mathcal G_t(\gamma)
=
\{0,F_t(0;\gamma),\ldots,F_t(E_t;\gamma)=1\}
\]
and its mesh
\begin{equation}
\Delta_t(\gamma)=\max_{0\le a\le E_t}m_t(a;\gamma).
\label{eq:mesh}
\end{equation}
The observed discrete process reveals latent probabilities only on products of
these grids.

\begin{assumption}[Refining grids and block identification]
\label{ass:grid}
\begin{enumerate}[label=(\roman*)]
\item the marginal parameter $\gamma_0$ is identified from the
exposure-conditioned one-dimensional distributions;
\item for each fixed block length $d\ge2$, there exists a sequence of consecutive
blocks $B_k$ such that
\[
\max_{t\in B_k}\Delta_t(\gamma_0)\longrightarrow0;
\]
\item the adjusted latent $d$-dimensional CDF is continuous;
\item within the specified parametric MAGMAR family, equality of all finite block
distributions implies equality of the dependence parameter.
\end{enumerate}
\end{assumption}

\begin{theorem}[Parametric identification under refining exposure grids]
\label{thm:grid}
Under Assumption~\ref{ass:grid}, if two parameter vectors
$\xi_1=(\eta_1,\gamma_1)$ and $\xi_2=(\eta_2,\gamma_2)$ generate the same
finite-dimensional distributions for the observed count process conditional on the
refining exposure blocks, then $\xi_1=\xi_2$.
\end{theorem}

\begin{proof}
Equality of the one-dimensional count distributions identifies
$\gamma_1=\gamma_2=\gamma_0$. For a fixed block length $d$, equality of observed
block distributions implies equality of adjusted latent block CDFs at every point
of the corresponding product probability grid. By the refining-mesh condition, the
union of these grids is dense in $[0,1]^d$. Continuity extends equality from the
grids to all of $[0,1]^d$, and block identification of the parametric MAGMAR family
then yields $\eta_1=\eta_2$.
\end{proof}

The theorem is a population statement. In the observed sample, exposure equals one
for many early years and later rises above 300, so grid resolution changes markedly
through time. Section~\ref{subsec:finite-sample} evaluates that finite-sample
resolution directly through the empirical mesh, profile objectives, low-exposure
subsamples, and parameter recovery.

\section{Simulated likelihood and validation design}
\label{sec:estimation}

\subsection{Guided sequential Monte Carlo likelihood}

For the static margin, the parameter vector is
\[
\xi_{\mathrm{MAGMAR}}=(\beta_0,\log\kappa,\phi,\theta)',
\qquad
\xi_{\mathrm{MAG}}=(\beta_0,\log\kappa,\theta)'.
\]
For the dynamic margin, $\beta_1$ is added. For each candidate parameter vector,
the likelihood calculation proceeds as follows.

\begin{enumerate}[leftmargin=2.2em]
\item \textbf{Count intervals.} Construct the beta-binomial intervals
$I_t(A_t;\gamma)$ under the candidate static or dynamic margin.
\item \textbf{Stationary marginal map.} For MAGMAR, simulate a long stationary raw
path at the candidate $(\phi,\theta)$ and construct numerical maps
$\widehat\Psi_\eta$ and $\widehat\Psi_\eta^{-1}$.
\item \textbf{Raw-state intervals.} Map the count intervals back through
$\widehat\Psi_\eta^{-1}$. For MAG, whose raw construction is uniform in the
implementation, the count intervals are used directly.
\item \textbf{Initialization.} The first count contributes its fitted marginal
probability. A stationary simulation pool is conditioned on the first raw-state
interval to initialize particles for $(U_1,W_1)$.
\item \textbf{Sequential integration.} For each subsequent year, particles propose
$U_t$ uniformly over the current raw-state interval. Importance weights equal the
interval width times the conditional density in \eqref{eq:magmar-density}. Weights
are normalized and systematic resampling propagates particles and innovation states.
\item \textbf{Optimization.} The sum of sequential log normalizing constants is the
simulated log likelihood. Common random numbers are held fixed across parameter
evaluations within a run. BFGS optimization is used with multiple starts.
\end{enumerate}

The saved publication fits reported in Section~\ref{sec:results} use 6,000 particles
and an 80,000-state stationary pool. The lighter validation exercises use 2,000
particles and a 20,000-state pool unless stated otherwise, with stationary burn-in
3,000, BFGS maximum iteration count 250, and relative tolerance $2\times10^{-5}$.
Direction-specific and influential-year refits use 1,500 particles and 16,000-state
pools.


\subsection{Implementation and Monte Carlo checks}

The implementation contains four direct unit checks. The Gaussian $h$-function and
its inverse agree to machine precision over several dependence parameters. Simulated
MAG states are marginally uniform, and the stationary-adjusted MAGMAR states are
uniform after applying the estimated stationary map. Under MAG independence
($\theta=0$), the interval likelihood agrees with the product of marginal interval
masses to numerical precision. A $T=2$ brute-force integration check agrees with the
SMC likelihood within $1.61\times10^{-3}$.

Monte Carlo stability is assessed in two ways. First, the fitted dynamic MAG and
MAGMAR parameters are re-evaluated on 12 independent streams for each combination of
particle counts $2{,}000$, $4{,}000$, $8{,}000$, and $12{,}000$ and stationary-pool
sizes $20{,}000$, $50{,}000$, and $100{,}000$. Second, four independent full refits
at the default validation configuration quantify simulation-induced variation in
both likelihood differences and parameter estimates.

\subsection{Profiles, bootstrap, and finite-sample recovery}

The dynamic-margin specification is used for the main finite-sample validation
because it directly addresses the strongest alternative explanation for apparent
copula persistence. Independent-stream nuisance profiles are computed for $\phi$ and
$\theta$. A parametric bootstrap simulates count series under fitted dynamic MAG on
the actual exposure path and refits MAG and MAGMAR. Parameter recovery simulates
from the reference dynamic MAGMAR fit on the same exposure path and re-estimates all
parameters.

The bootstrap is used as a numerical calibration exercise rather than as a source of
high-precision tail probabilities. With only 49 replications, its Monte Carlo
resolution is coarse. In addition, the stored run contains one obvious null-fit
optimization failure, producing an LR statistic above 1,500, and several negative
computed LR values caused by the alternative optimizer not fully recovering the
nested null solution. These features are reported transparently in
Section~\ref{subsec:dynamic-results}; they do not support a conventional Wilks-style
significance claim.

\section{Empirical results}
\label{sec:results}

\subsection{Overdispersion and marginal persistence}
\label{subsec:margin-results}

The static beta-binomial margin has fitted precision
$\widehat\kappa=44.70$ and substantially improves on the exposure-only binomial
margin used in preliminary analysis. More importantly for the dependence analysis,
its Pearson residuals retain strong serial structure: the Ljung--Box statistic at
lag 10 is 102.43, with $p\approx1.8\times10^{-17}$. The dynamic beta-binomial margin
changes this diagnostic sharply. With
\[
\widehat\beta_0=-2.205,
\qquad
\widehat\beta_1=4.105,
\qquad
\widehat\kappa=104.37,
\]
the Ljung--Box statistic falls to 10.04 and the nominal $p$-value rises to 0.437.
The dynamic count margin therefore absorbs essentially all residual linear serial
correlation detected by this diagnostic.

\begin{table}[t]
\centering
\caption{Beta-binomial margin diagnostics. The static fit uses 1963--2025; the dynamic fit begins in 1964 because it conditions on the lagged migration rate.}
\label{tab:margin-results}
\begin{tabular}{lrrrrrr}
\toprule
Margin & $n$ & $\log L$ & AIC & $\widehat\kappa$ & Residual SD & LB(10) $p$ \\
\midrule
Static beta-binomial  & 63 & -189.792 & 383.584 & 44.703  & 1.019 & $<10^{-16}$ \\
Dynamic beta-binomial & 62 & -174.015 & 354.029 & 104.373 & 1.125 & 0.437 \\
\bottomrule
\end{tabular}
\end{table}

Because the two rows use different starting years, the likelihood criteria in
Table~\ref{tab:margin-results} are descriptive rather than a formal nested test. The
serial-diagnostic change is the central point: a static observation mean leaves a
large predictable component in the counts, whereas a one-lag migration-rate term
removes it.

\subsection{Static-margin interval likelihood strongly favors MAGMAR}
\label{subsec:static-results}

Table~\ref{tab:static-primary} reports the saved publication fits under the static
beta-binomial margin. Gaussian MAGMAR(1,1) has simulated observed-count log
likelihood $-170.688$, compared with $-178.876$ for Gaussian MAG(1). The descriptive
nested statistic is
\[
2(-170.688+178.876)=16.376.
\]
AIC falls from 363.75 to 349.38 and BIC from 370.18 to 357.95. The MAGMAR fit has a
large positive autoregressive parameter, $\widehat\phi=0.879$, and a negative
moving-aggregate parameter, $\widehat\theta=-0.530$.

\begin{table}[t]
\centering
\caption{Static beta-binomial observed-count interval-SML fits.}
\label{tab:static-primary}
\begin{tabular}{lrrrrrrr}
\toprule
Model & $k$ & $\log L_{\mathrm{SML}}$ & AIC & BIC & $\widehat\kappa$ & $\widehat\phi$ & $\widehat\theta$ \\
\midrule
MAG(1)      & 3 & -178.876 & 363.752 & 370.181 & 63.238 & ---   & 0.404 \\
MAGMAR(1,1) & 4 & -170.688 & 349.376 & 357.949 & 89.026 & 0.879 & -0.530 \\
\bottomrule
\end{tabular}
\end{table}

The static-margin result is clear as a statement about that observation model. It is
not sufficient by itself for a structural interpretation of latent persistence,
because the static beta-binomial residuals remain serially correlated. This is why
the dynamic margin is treated as the central specification sensitivity rather than
as an optional appendix check.

\subsection{Dynamic margin: a much weaker and less sharply identified MAGMAR increment}
\label{subsec:dynamic-results}

On the common 1964--2025 sample, the reference dynamic-margin MAG fit has log
likelihood $-173.888$. The reference MAGMAR fit improves this to $-173.119$, a gain
of only 0.769 log-likelihood units and an LR statistic of 1.537. With one additional
parameter, MAGMAR has slightly higher AIC and BIC than MAG. The corresponding
reference dependence estimates are
\[
\widehat\phi=0.684,
\qquad
\widehat\theta=-0.502.
\]

\begin{table}[t]
\centering
\caption{Reference dynamic beta-binomial interval-SML fits on 1964--2025.}
\label{tab:dynamic-primary}
\small
\begin{tabular}{lrrrrrrrr}
\toprule
Model & $k$ & $\log L_{\mathrm{SML}}$ & AIC & BIC & $\widehat\beta_1$ & $\widehat\kappa$ & $\widehat\phi$ & $\widehat\theta$ \\
\midrule
MAG(1)      & 4 & -173.888 & 355.776 & 364.284 & 3.645 & 104.194 & ---   & 0.079 \\
MAGMAR(1,1) & 5 & -173.119 & 356.238 & 366.874 & 4.308 & 92.008  & 0.684 & -0.502 \\
\bottomrule
\end{tabular}
\end{table}

The reference pair should not be read as establishing a unique global maximum for
the dynamic MAGMAR objective. Independent nuisance profiles reveal a higher-$\phi$
region. Averaging over two independent likelihood streams, the profile log
likelihood rises from $-172.94$ at $\phi=0.684$ to $-171.35$ at $\phi=0.90$.
Higher-precision full refits display the same behavior: the MAG fit remains near
$-173.88$, whereas the MAGMAR solution can move toward $\phi\approx0.93$ and
$\theta\approx-0.63$. This is evidence of a broad or multimodal dependence surface,
not a basis for selecting one numerical mode as a precise estimator.


The 49-replication parametric bootstrap around the reference dynamic MAG fit is also
non-decisive. The observed LR is 1.537. Retaining all stored replications gives a
plus-one upper-tail probability of 0.12. One replication is an obvious numerical
failure: the null MAG refit lands at log likelihood $-943.7$ while the MAGMAR refit
is near $-165.6$, producing an LR above 1,500 despite a nominal convergence flag.
Treating this run as invalid leaves 48 replications and a plus-one exceedance rate of
approximately 0.10. Several negative computed LRs further indicate optimizer noise
in the nested comparison. The bootstrap therefore supports the qualitative
conclusion that the reference dynamic MAGMAR increment is not strongly significant,
but it is not used for precise tail inference.

\subsection{Monte Carlo stability and profile geometry}
\label{subsec:numerical-results}

At the fixed reference parameters, paired independent-stream evaluations are much
more stable than the full refits. Across 12 particle/pool configurations, the mean
MAGMAR--MAG log-likelihood difference ranges from 0.395 to 0.690. The share of
replications with a positive difference ranges from 0.833 to 1.000, and the mean
MAGMAR median ESS is essentially equal to the particle count in every configuration.
Thus the particle filter itself is not collapsing around the reference solution.

\begin{table}[t]
\centering
\caption{Selected paired Monte Carlo diagnostics at the reference dynamic-margin parameters. Each row uses 12 independent streams.}
\label{tab:paired-mc}
\begin{tabular}{rrrrrr}
\toprule
Particles & Pool & Mean $\Delta\ell$ & SD $\Delta\ell$ & MCSE(mean) & Share $\Delta\ell>0$ \\
\midrule
2,000  & 20,000  & 0.690 & 0.493 & 0.142 & 0.917 \\
4,000  & 100,000 & 0.582 & 0.115 & 0.033 & 1.000 \\
8,000  & 100,000 & 0.592 & 0.159 & 0.046 & 1.000 \\
12,000 & 100,000 & 0.497 & 0.126 & 0.036 & 1.000 \\
\bottomrule
\end{tabular}
\end{table}

Four independent full refits at the default validation configuration give mean
$\Delta\ell=1.316$ with MCSE 0.235. The mean dependence estimates are
$\widehat\phi=0.691$ and $\widehat\theta=-0.512$, with MCSEs 0.012 and 0.023,
respectively. The margin estimates are similarly stable. The contrast between stable
local refits and a profile that improves at high $\phi$ is characteristic of a
multimodal or ridge-like objective: simulation noise is manageable locally, while
global separation of the dependence components is weaker.

\subsection{Finite-sample grid resolution and parameter recovery}
\label{subsec:finite-sample}

The empirical probability-grid mesh confirms the role of changing exposure in
Theorem~\ref{thm:grid}. Under the reference dynamic MAGMAR margin, the largest cell
width has median 0.077, minimum 0.021, and maximum 0.904. The correlation between
exposure and the maximum grid-cell width is $-0.716$. The very coarse early years are
therefore fundamentally less informative about latent probabilities than the later
high-exposure years.

\begin{figure}[H]
\centering
\includegraphics[width=0.88\textwidth]{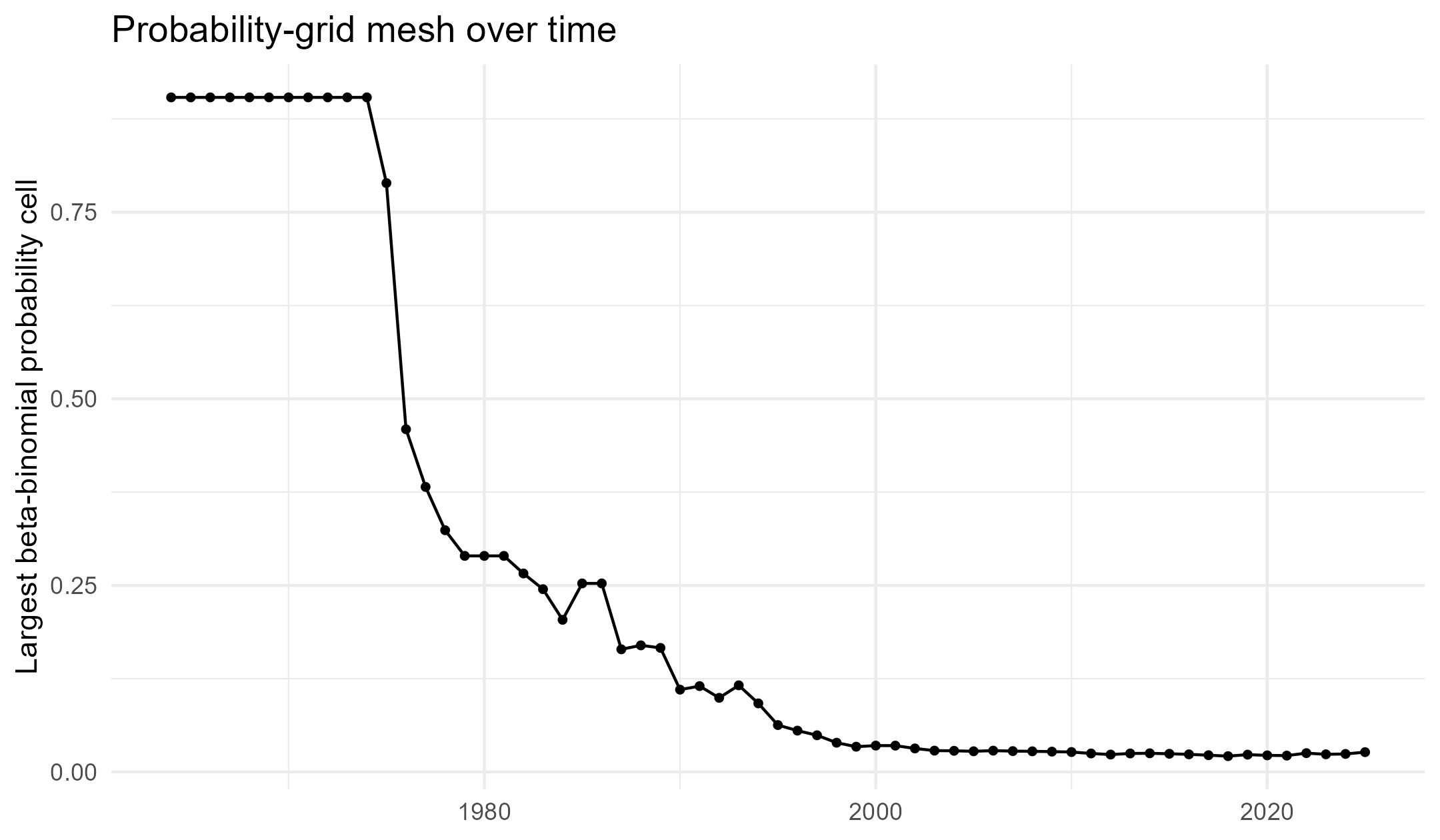}
\caption{Maximum beta-binomial probability-grid cell width by year under the reference dynamic-margin fit. Larger exposure produces a finer observable probability grid.}
\label{fig:grid-mesh}
\end{figure}

Parameter-recovery experiments make the finite-sample implication explicit. Twenty
series are simulated from the reference dynamic MAGMAR fit on the actual exposure
path and then re-estimated. All runs return convergence flags and none are classified
as boundary solutions. The margin parameters recover well: biases are 0.006 for
$\beta_0$, $-0.004$ for $\beta_1$, and $-1.80$ for $\kappa$. Dependence recovery is
much weaker. The bias and RMSE are $-0.329$ and 0.470 for $\phi$, and 0.324 and
0.475 for $\theta$.

\begin{table}[t]
\centering
\caption{Finite-sample parameter recovery on the actual exposure path, 20 replications.}
\label{tab:recovery}
\begin{tabular}{lrrrr}
\toprule
Parameter & True value & Bias & RMSE & 2.5\%--97.5\% of estimates \\
\midrule
$\phi$   & 0.684 & -0.329 & 0.470 & $[-0.133,\ 0.793]$ \\
$\theta$ & -0.502 & 0.324 & 0.475 & $[-0.763,\ 0.341]$ \\
$\beta_0$ & -2.241 & 0.006 & --- & --- \\
$\beta_1$ & 4.308 & -0.004 & --- & --- \\
$\kappa$  & 92.008 & -1.799 & --- & --- \\
\bottomrule
\end{tabular}
\end{table}

The pattern is informative: the exposure-conditioned mean is recoverable at this
sample size, while the decomposition of serial dependence into $\phi$ and $\theta$
is not sharply resolved. This is consistent with the profile geometry and with the
change in model ranking after adding the lagged count rate to the margin.

Low-exposure subsamples reinforce the same conclusion. Starting the dynamic analysis
when exposure first reaches 25, 50, or 100 produces descriptive MAGMAR--MAG LR
statistics of approximately $-0.01$, 2.84, and $-4.35$, respectively. These signs
are not interpreted literally as nested LR statistics because the simulated
optimization can return a lower alternative objective. The relevant result is that
there is no stable MAGMAR advantage across exposure thresholds.

\subsection{Influential years and directional decomposition}

The largest dynamic beta-binomial Pearson residuals occur in 1975 and 1974, when
exposure is extremely small. Masking either year, or both together, leaves the
reference MAGMAR dependence parameters close to the full-sample local solution.
Masking both gives $\widehat\phi=0.684$ and $\widehat\theta=-0.502$ and a descriptive
MAGMAR--MAG LR of 1.34. The local dependence pattern is therefore not generated by a
single exceptional early observation.

\begin{figure}[H]
\centering
\includegraphics[width=0.88\textwidth]{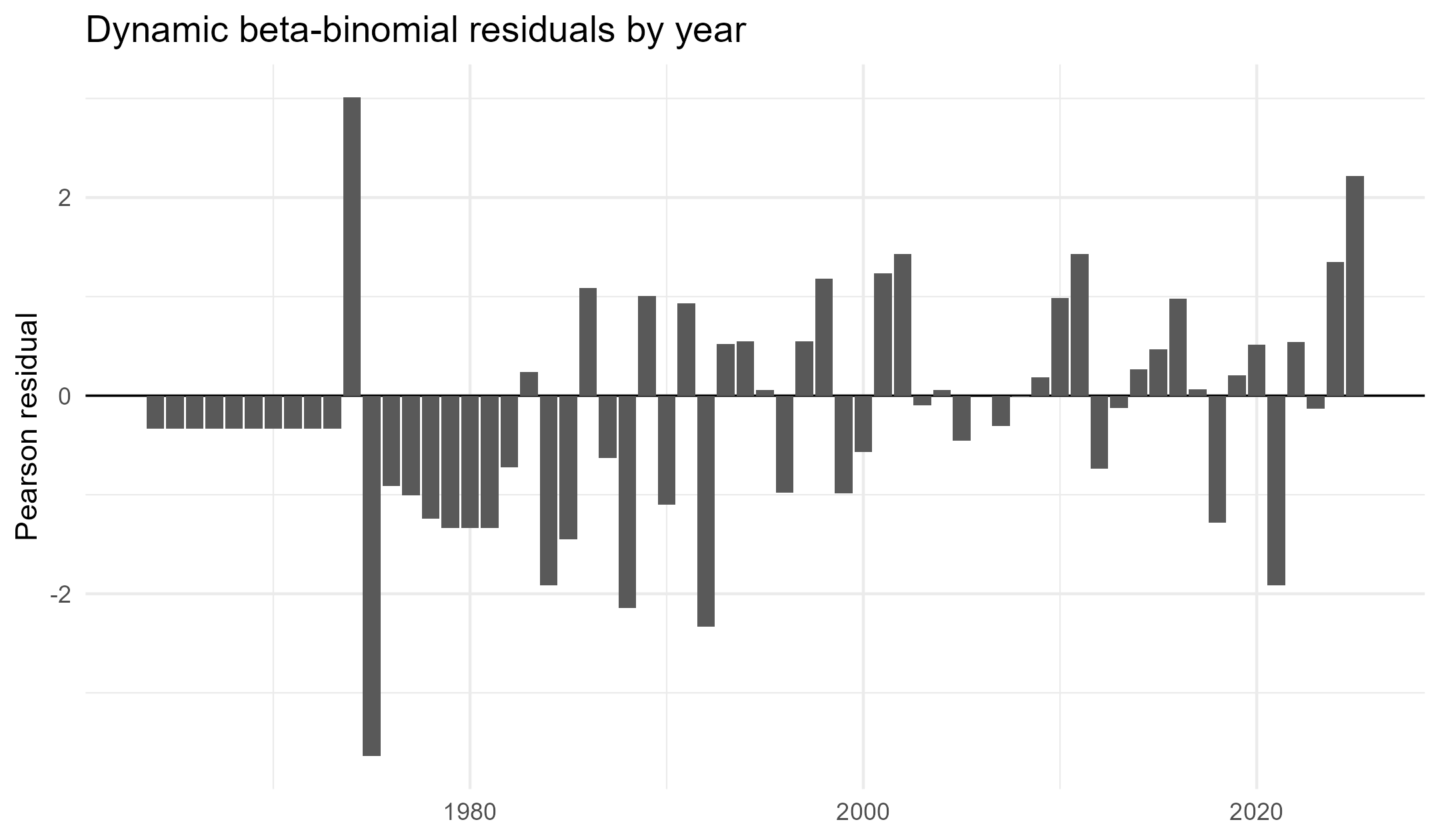}
\caption{Pearson residuals from the dynamic beta-binomial count margin. The two largest absolute residuals occur in 1975 and 1974, when exposure is minimal.}
\label{fig:influence}
\end{figure}

Direction-specific interval fits are exploratory because splitting total migration
activity into upgrades and downgrades reduces the effective information further.
Under static beta-binomial margins, MAG outperforms MAGMAR for downgrades, while
MAGMAR improves the upgrade log likelihood by 1.64 units. The separate fits also
push the MAG moving-aggregate parameter above 0.88 and produce very low minimum ESS
in some years. These results suggest asymmetry between upgrade and downgrade
activity, but they are not used for the primary structural comparison.

\section{Discussion}
\label{sec:discussion}

The central methodological result is that the discrete observation layer matters.
A sovereign migration count maps to an interval of latent probability values, and
MAGMAR's stationary marginal adjustment changes that interval as the dependence
parameters change. Integrating over the interval is therefore part of the model,
not a cosmetic correction to a continuous transformed-series likelihood.

The central empirical result is a separation problem between marginal and copula
persistence. If the beta-binomial migration probability is static, the observed
count likelihood strongly favors MAGMAR. Once the migration probability responds to
the previous year's migration rate, the count margin itself explains nearly all of
the residual serial correlation detected by a Ljung--Box diagnostic. The remaining
increment from the reference MAGMAR solution is small, and the likelihood surface
contains a high-$\phi$ region that is not well separated from the moderate-$\phi$
solution. In short annual data, richer marginal dynamics and richer latent
dependence can substitute for one another.

This interaction explains why a single model-selection statistic is not the right
summary of the application. Under the static margin, AIC and BIC strongly favor
MAGMAR. Under the reference dynamic local solution, they favor MAG. Independent
profiles and higher-precision refits then reveal that the dynamic MAGMAR objective
can improve again at large $\phi$. The numerical diagnostics are therefore part of
the empirical result: they show that the data contain evidence of serial structure,
but do not sharply allocate that structure between a lagged conditional mean and two
copula dependence channels.

The parameter-recovery experiment provides the clearest finite-sample interpretation.
The dynamic beta-binomial coefficients and precision parameter are recovered with
small bias, whereas the dependence parameters have RMSEs near 0.47 and broad
empirical ranges. This distinction is consistent with the observed probability-grid
mesh. Later high-exposure years provide fine count intervals, but the series still
contains only 62 dynamic observations, and the early part of the sample contributes
very coarse probability information.

For credit-risk applications, the framework is most useful when the research
question concerns aggregate migration activity rather than individual transition
matrices. It complements transition-matrix and intensity models by providing a way
to model persistent annual variation in the total amount of rating movement after
conditioning on exposure. The analysis also gives a practical specification lesson:
a latent dependence model should be compared against a dynamic count margin before
persistence is assigned to the copula recursion.

\section{Conclusion}
\label{sec:conclusion}

This paper develops a discrete interval-likelihood framework for applying MAG and
MAGMAR copula time-series models to sovereign rating-migration counts with changing
exposure. The observed count identifies a CDF interval rather than a unique PIT
value, and the MAGMAR stationary marginal transformation makes the corresponding raw
latent interval parameter-dependent. A guided SMC algorithm integrates over these
intervals directly.

In the sovereign application, a static beta-binomial margin produces strong evidence
in favor of MAGMAR relative to MAG. That result changes once marginal persistence is
modeled explicitly. A beta-binomial mean driven by the lagged migration rate removes
the residual serial correlation left by the static margin, and the reference dynamic
MAGMAR increment becomes small. At the same time, independent profiles and
higher-precision refits reveal a high-autoregressive-dependence region, while
parameter recovery shows that $\phi$ and $\theta$ are weakly resolved in samples of
this length.

The broader implication is twofold. First, discrete copula time-series models should
be estimated on the probability intervals implied by the observations rather than
on arbitrary randomized representatives of those intervals. Second, in short count
series, evidence for latent serial dependence must be interpreted jointly with the
dynamics allowed in the observation margin. For sovereign migration counts, the
interval likelihood identifies persistent structure, but the available sample does
not support a sharp decomposition of that persistence into marginal and MAGMAR
components.

\appendix

\section{Implementation unit tests}
\label{app:unit-tests}

\begin{table}[H]
\centering
\caption{Implementation checks used in the current validation run.}
\label{tab:unit-tests}
\small
\begin{tabular}{p{0.42\textwidth}p{0.25\textwidth}rr}
\toprule
Test & Setting & Statistic & Pass \\
\midrule
$h$-inverse identity & $\rho=-0.5$ & $4.44\times10^{-16}$ & yes \\
$h$-inverse identity & $\rho=0$ & $2.22\times10^{-16}$ & yes \\
$h$-inverse identity & $\rho=0.25$ & $2.50\times10^{-16}$ & yes \\
$h$-inverse identity & $\rho=0.6$ & $6.66\times10^{-16}$ & yes \\
MAG marginal uniformity & $\theta=0.35$ & 0.00334 & yes \\
Adjusted MAGMAR uniformity & $\phi=0.30,\theta=0.20$ & 0.00420 & yes \\
Independence interval identity & MAG $\theta=0$ & $8.88\times10^{-16}$ & yes \\
$T=2$ brute force vs SMC & MAG $\theta=0.35$ & 0.00161 & yes \\
\bottomrule
\end{tabular}
\end{table}

\section{Full paired Monte Carlo grid}
\label{app:paired-mc}

\begin{table}[H]
\centering
\caption{Paired independent-stream evaluation of the dynamic reference fits. Each configuration has 12 valid replications.}
\label{tab:paired-full}
\small
\begin{tabular}{rrrrrr}
\toprule
Particles & Pool & Mean $\Delta\ell$ & SD & MCSE(mean) & Share positive \\
\midrule
2,000  & 20,000  & 0.690 & 0.493 & 0.142 & 0.917 \\
2,000  & 50,000  & 0.395 & 0.196 & 0.057 & 1.000 \\
2,000  & 100,000 & 0.567 & 0.188 & 0.054 & 1.000 \\
4,000  & 20,000  & 0.477 & 0.407 & 0.117 & 0.833 \\
4,000  & 50,000  & 0.633 & 0.218 & 0.063 & 1.000 \\
4,000  & 100,000 & 0.582 & 0.115 & 0.033 & 1.000 \\
8,000  & 20,000  & 0.455 & 0.419 & 0.121 & 0.833 \\
8,000  & 50,000  & 0.607 & 0.314 & 0.091 & 1.000 \\
8,000  & 100,000 & 0.592 & 0.159 & 0.046 & 1.000 \\
12,000 & 20,000  & 0.592 & 0.392 & 0.113 & 0.917 \\
12,000 & 50,000  & 0.573 & 0.377 & 0.109 & 0.917 \\
12,000 & 100,000 & 0.497 & 0.126 & 0.036 & 1.000 \\
\bottomrule
\end{tabular}
\end{table}

\section{Independent-stream profile summaries}
\label{app:profiles}

\begin{table}[H]
\centering
\caption{Dynamic MAGMAR nuisance profiles averaged over two independent likelihood streams. Relative likelihoods are normalized separately within each fixed parameter.}
\label{tab:profiles}
\small
\begin{tabular}{lrrr}
\toprule
Fixed parameter & Fixed value & Mean profile $\log L$ & Relative $\log L$ \\
\midrule
$\phi$ & 0.000 & -176.616 & -5.261 \\
$\phi$ & 0.384 & -174.042 & -2.688 \\
$\phi$ & 0.513 & -174.607 & -3.253 \\
$\phi$ & 0.642 & -173.455 & -2.101 \\
$\phi$ & 0.684 & -172.939 & -1.585 \\
$\phi$ & 0.771 & -172.629 & -1.275 \\
$\phi$ & 0.900 & -171.354 & 0.000 \\
\midrule
$\theta$ & -0.800 & -173.642 & -0.856 \\
$\theta$ & -0.651 & -178.855 & -6.070 \\
$\theta$ & -0.502 & -172.785 & 0.000 \\
$\theta$ & -0.501 & -173.034 & -0.249 \\
$\theta$ & -0.352 & -175.804 & -3.019 \\
$\theta$ & -0.202 & -174.280 & -1.495 \\
\bottomrule
\end{tabular}
\end{table}

The large between-stream standard deviation at $\theta=-0.651$ and the monotone
improvement of the $\phi$ profile toward the top of the evaluated grid reinforce the
interpretation of a broad and numerically challenging dependence surface.

\section{Low-exposure and influential-year sensitivity}
\label{app:sensitivity}

\begin{table}[H]
\centering
\caption{Dynamic-margin low-exposure sensitivity. LR values are descriptive because simulated optimization can violate the exact nesting inequality.}
\label{tab:low-exposure}
\small
\begin{tabular}{rrrrrrrr}
\toprule
Exposure threshold & Start & $n$ & $\ell_{MAG}$ & $\ell_{MAGMAR}$ & LR & $\widehat\phi$ & $\widehat\theta$ \\
\midrule
25  & 1988 & 38 & -143.719 & -143.724 & -0.008 & 0.625 & -0.688 \\
50  & 1991 & 35 & -132.709 & -131.288 & 2.840 & 0.574 & -0.703 \\
100 & 1996 & 30 & -114.745 & -116.919 & -4.348 & 0.349 & -0.890 \\
\bottomrule
\end{tabular}
\end{table}

\begin{table}[H]
\centering
\caption{Influential-year sensitivity under the dynamic beta-binomial margin.}
\label{tab:influence}
\begin{tabular}{lrrrrr}
\toprule
Masked year(s) & $\ell_{MAG}$ & $\ell_{MAGMAR}$ & LR & $\widehat\phi$ & $\widehat\theta$ \\
\midrule
1975 & -170.099 & -168.964 & 2.271 & 0.700 & -0.480 \\
1974 & -171.175 & -168.975 & 4.399 & 0.678 & -0.484 \\
1975, 1974 & -167.465 & -166.795 & 1.339 & 0.684 & -0.502 \\
\bottomrule
\end{tabular}
\end{table}

\section{Direction-specific exploratory fits}
\label{app:direction}

\begin{table}[H]
\centering
\caption{Direction-specific interval fits under static beta-binomial margins.}
\label{tab:direction}
\small
\begin{tabular}{llrrrrr}
\toprule
Direction & Model & $\log L$ & AIC & BIC & $\widehat\phi$ & $\widehat\theta$ \\
\midrule
Downgrade & MAG    & -152.604 & 311.208 & 317.637 & ---   & 0.941 \\
Downgrade & MAGMAR & -154.218 & 316.436 & 325.009 & 0.199 & 0.148 \\
Upgrade   & MAG    & -159.072 & 324.145 & 330.574 & ---   & 0.884 \\
Upgrade   & MAGMAR & -157.428 & 322.855 & 331.428 & 0.554 & -0.229 \\
\bottomrule
\end{tabular}
\end{table}





\section*{Declarations of Interest}

The author reports no conflicts of interest and is solely responsible for the
content and writing of the paper.

\end{document}